\documentclass[12pt,natbib]{article}
\usepackage[round,colon,authoryear]{natbib}
\usepackage{bibentry}
\RequirePackage[OT1]{fontenc}
\RequirePackage{amsthm,amsmath,amsfonts,amssymb,{mathtools}}
\RequirePackage[colorlinks]{hyperref}
\usepackage{bbm}
\hypersetup{
 colorlinks=false,
 citecolor=Violet,
 linkcolor=Red,
 urlcolor=Blue}
\RequirePackage{epsfig, graphicx, color}
\RequirePackage{latexsym}
\RequirePackage{psfrag}
\usepackage{multirow}
\usepackage{diagbox}
\usepackage{fullpage}
\usepackage{epsf,epsfig,subfigure}
\usepackage{graphics,graphicx}
\usepackage{enumerate}
\usepackage{epsfig}

\usepackage{amsmath, amsfonts, amssymb,, mathtools, subfigure,bbm}
\usepackage[figuresright]{rotating}

\usepackage[ruled, vlined, lined, commentsnumbered]{algorithm2e}
\usepackage{multirow}
\usepackage{hyperref}

\newtheorem{lemma}{Lemma}[section]
\newtheorem{proposition}{Proposition}[section]
\newtheorem{thm}{Theorem}[section]

\def\Var{\textrm{Var}} % the symbol Var for covariance used the sans serif letter

\def\E{\mathbb{E}}
\def\P{\mathbb{P}}

\def\text#1{\mbox{\rm #1}}

\def\bd{\boldsymbol}

\newcommand{\argmin}{\mathop{\rm argmin}}

\newcommand{\R}{\mathbb{R}}

\begin{document}

\title{Clustering High-dimensional Data via Feature Selection}

\author{Tianqi Liu$^{1}$,
Yu Lu$^{2}$, Biqing Zhu$^{3}$, and Hongyu Zhao$^{3}$ \\
$^{1}$Google Research, $^{2}$Two Sigma, $^{3}$Yale University
}
\maketitle
\begin{abstract}
High-dimensional clustering analysis is a challenging problem in statistics and machine learning, with broad applications such as the analysis of microarray data and RNA-seq data. In this paper, we propose a new clustering procedure called Spectral Clustering with Feature Selection (SC-FS), where we first obtain an initial estimate of labels via spectral clustering, then select a small fraction of features with the largest R-squared with these labels, i.e., the proportion of variation explained by group labels, and conduct clustering again using selected features. Under mild conditions, we prove that the proposed method identifies all informative features with high probability and achieves minimax optimal clustering error rate for the sparse Gaussian mixture model.  Applications of SC-FS to four real world data sets demonstrate its usefulness in clustering high-dimensional data. 
\end{abstract}

\section{Introduction \label{sec:intro}}
Consider a high-dimensional clustering problem, where we observe $n$ vectors $Y_i\in \R^p, i=1,2,\cdots,n,$ from $k$ clusters with $p>n$.  The task is to group these observations into $k$ clusters such that the observations within the same cluster are more similar to each other than those from different ones. 

Several statistical methods have been proposed to tackle the high-dimensional clustering problem \citep{pan2007penalized,guo2010pairwise,krishnamurthy2011high,witten2012framework,wu2016new,jin2016influential,song2011fast,dash2000feature,xing2001cliff,chakraborty2020entropy, liu2022characterizing, kriegel2009clustering}. A popular choice is to add regularization to encourage sparsity: \citet{pan2007penalized} added $L_1$ penalty on the cluster mean of each feature, \citet{guo2010pairwise} used pairwise group-fusion penalty to reduce the difference between different groups, \citet{witten2012framework} developed sparse $k$-means and sparse hierarchical clustering via sparse weighted loss of each feature. While the numerical results of these methods were promising, there was no theoretical justification of these methods. Besides enforcing sparsity, several works propose to cluster on latent space via matrix factorization, tensor decomposition or random projection \citep{rohe2011spectral,  liu2022characterizing, kriegel2009clustering, fern2003random}. Another way to address the high dimensionality is through feature selection \citep{chormunge2018correlation, xing2001cliff, dash2000feature}. High-dimension feature screening has been well studied under supervised learning \citep{fan2008sure, fan2009ultrahigh, balasubramanian2013ultrahigh, liu2016ultrahigh}. For unsupervised learning, \citet{jin2016influential} proposed Influential Features PCA (IFPCA), in which they considered selecting influential features by Kolmogorov-Smirnov (KS) scores. They obtained consistency clustering under the sparse Gaussian mixture model.  However,  their convergence rate is far from the optimal exponentially small clustering error. And the computational cost of calculating KS scores is relatively high. 

In this paper, we propose a computationally efficient and provably optimal method to solve high-dimensional clustering problem. Our approach is motivated from recent progress in single cell RNA sequencing (scRNA-seq) data analysis \citep{patel2014single, zeisel2015cell, chen2018viper, zamanighomi2018unsupervised, su2021accurate, hao2021integrated}. When clustering cell types from the same tissue, it is natural to assume that most of the genes are not differentially expressed and only cell-type specific genes can be informative on identifying cell types. We can use pseudo labeling techniques \citep{lee2013pseudo} and select informative features on the psuedo labels. Formally, our approach consists of three stages, in which we first obtain an initial estimate of the labels by spectral clustering, and we then select informative features using $R$-squared of univariate regressions on estimated labels, and finally run spectral clustering with Lloyd's iterations on the selected features. Under mild conditions, we show that the proposed algorithm can successfully identify all informative features. More specifically, given any consistent initial estimate of labels, the second stage of our algorithm selects all informative features with over-whelming probability under the sparse Gaussian mixture model. With those informative features, we are able to run Lloyd iterations in stage three to achieve the optimal mis-clustering rate \cite{lu2016statistical}. More specifically, we show that 

\begin{thm} \label{thm:intro} [Informal]
Under mild sample size and signal-to-noise ratio conditions, our three-stage algorithm achieves an exponentially small mis-clustering rate, which is minimax optimal up to constant in the exponent, w.h.p.
\end{thm}

We refer the readers to Theorem \ref{thm:main} in Section \ref{sec:theo} for the exact conditions we need. Another contribution of our analysis is to derive a faster convergence rate of spectral clustering. Inspired by the recent perturbation results for singular sub-spaces \cite{cai2016rate}, we improve the error rate of spectral clustering from $O(\sqrt{p/n})$ to $O((p/n)^{1/4})$ when $p>n$. Our proposed method provides a new way to efficiently characterize sub-populations in a heterogeneous dataset, identify informative genes, and gain biological insights from high-dimensional datasets such as scRNA-seq data.

The rest of the paper is organized as follows. Section \ref{sec:meth} introduces (SC-FS) methodology. Theoretical results are provided in Section \ref{sec:theo}. Section \ref{sec:experiments} reports the results from numerical studies, including synthetic data study and four real data applications. Finally, we conclude the paper with some remarks and discussions in Section \ref{sec:dis}.

\section{Methodology}
\label{sec:meth}
In this section, we formally introduce the sparse Gaussian mixture model considered in the paper. Then we present the three stages of our SC-FS algorithm. 
\subsection{Sparse Gaussian Mixture Model}
Suppose there are $k$ clusters with center matrix $B\in\mathbb{R}^{k\times p}$, with rows $B_{1*}, \cdots, B_{k*} \in \mathbb{R}^{p}$ being centers of clusters. We observe independent samples from the following Gaussian mixture model.
\begin{equation} \label{eq:model}
Y_i = B_{z_i*} + W_i, \quad i=1,2,\cdots,n
\end{equation}
where $\{W_{i}\}$ are independent sub-Gaussian random vectors satisfying $$\E \exp \left( \gamma^T W_i \right) \le \exp(\|\gamma\|^2 \sigma^2/2)$$ for any $\gamma \in \mathbb{R}^d$ and $z_i \in [k]$ is the cluster label of the $i$th sample. Let $[p]$ denote the set $\{1,2,\cdots,p\}$.  For $j \in [p]$, let $\sigma_j^2$ be the marginal variance of the $j$-th feature. Here the variances for different features are not necessarily the same. For any subset of $A\subseteq[p]$, denote $\sigma_A=\max_{i \in A} \sigma_i$.  Let $T_a$ be the $a$-th cluster, i.e., $T_a = \{ i \in [n], z_i = a \}$ for $a \in [k]$. 

As we discussed in the introduction, there are many non-informative features under the ``large $p$, small $n$'' scenario. We refer to a feature as non-informative if its within-cluster means are the same across different clusters. Suppose there are $s$ informative features. Then the centers $B_1, \cdots, B_k$ only differ at $s$ coordinates.  Without loss of generality, we assume there is a subset $S \subset [p]$ with cardinality $s$ such that $B_{ij} = 0$ for all $i \in [k]$ and $j \in S^c$, where $B_{ij}$ is the $j$-th entry of center $B_i$. In practice, we can achieve this by centering and standard scaling each column.

\subsection{Algorithm}
In this section, we present our algorithm for clustering sparse Gaussian mixture data. The algorithm consists of three stages. In the first stage, we obtain an initial estimator of the labels by spectral clustering. Then we perform a feature selection step based on the initial label estimators. Finally, we run spectral clustering and Lloyd's algorithm on the selected features. 
\subsubsection{Stage 1: Spectral Clustering} \label{sec:spec}
In order to get a good initial estimator of the labels, we first perform de-noising via singular value decomposition (SVD), which preserves the cluster structure on the left eigenvectors under the noiseless case. More precisely, we can rewrite our model (\ref{eq:model}) as $Y = ZB + W$, where $$Z \in \mathcal{Z} =\left\{ A \in \{0,1\}^{n \times k},  \|A_{i*}\|_0 = 1, i \in [n] \right\}$$ is a membership matrix that has exactly one 1 in each row.  Then the SVD of the mean matrix $ZB$ has the following property. 
\begin{lemma} \label{lm:svdprop}
Let $U D V^T$ be the singular value decomposition of $ZB$, where $B$ is full rank. Then $U = ZQ$ with $Q \in \mathbb{R}^{k \times k}$ and $\|Q_{u*} - Q_{v*}\| = \sqrt{\frac{1}{n_u^*} + \frac{1}{n_v^*}}$ for all $1 \le u < v \le k$.  Moreover, $\sigma_{k}(ZB) \ge \sqrt{\alpha n} \sigma_{k}(B)$, where $\alpha n$ is the smallest cluster size.

\end{lemma}
This lemma is an immediate consequence of Lemma 2.1 in \cite{lei2013consistency} by noticing that the left singular vectors of $ZB$ are orthonormal eigenvectors of $ZBB^TZ^T$. Lemma \ref{lm:svdprop} implies that there are only $k$ different rows of $U$ and we can recover the cluster labels from it. Intuitively, when we have noisy observations of the $ZB$ matrix, $\widehat{U}$, the leading $k$ left singular vectors of sample matrix $Y$, should not differ from $U$ much. Since the rows of $U$ are well separated, we could run a distance-based clustering algorithm on the rows $\widehat{U}$ to estimate the labels. Theoretically, $k$-means problem is NP-hard and hence we use a polynomial-time approximation scheme of $k$-means. One possible choice is the $(1+\epsilon)$-approximate $k$-means algorithm proposed in \cite{kumar2004simple}. Another choice is the kmeans++ algorithm \cite{arthur2007k}. Although kmeans++ is only guaranteed to be a $(1+\log k)$-approximation in expectation,  it usually enjoys good performance in practice. 

\begin{algorithm}
\caption{Spectral Clustering} \vspace{0.05in}
\textbf{Input}: $Y_1, Y_2, \cdots, Y_n$. The number of clusters $k$. \vspace{0.05in} \\
\textbf{Output}: Estimated clusters $G_1, G_2, \dots, G_k$. 
\begin{itemize}
\item[1.] Compute $\widehat{U} \in \mathbb{R}^{n \times k}$ consisting of the leading $k$ left singular vectors (ordered in singular values) of $Y=[Y_1,\cdots, Y_n]^T$. 
\item[2.] Run $(1+\epsilon)$-approximation k-means on the rows of $\widehat{U}$, i.e. find $\widehat{Q} \in \mathbb{R}^{k \times k}$ and $\widehat{Z} \in \mathcal{Z}$ such that
\begin{equation} \label{eq:appkmeans}
\| \widehat{Z} \widehat{Q} - \widehat{U} \|_F^2 \le (1+\epsilon) \min_{\overline{Z} \in \mathcal{Z}, \overline{Q} \in \mathbb{R}^{k \times k}} \|\overline{Z} ~\overline{Q} - \widehat{U}\|_F^2
\end{equation}
\end{itemize}
\label{alg:spectralclustering}
\end{algorithm}

The above ideas are summarized in Algorithm \ref{alg:spectralclustering}.  We would like to remark that this spectral clustering algorithm is different from the popular one used in Gaussian mixture literature \cite{kumar2010clustering, awasthi2012improved, kannan2009spectral}, which runs clustering algorithm on the best rank $k$ projections of the data matrix $Y$. As we shall see in Section \ref{sec:specthm}, while these two algorithms theoretically work equally well for the low dimensional Gaussian mixture models, Algorithm \ref{alg:spectralclustering} is better for the high-dimensional sparse Gaussian mixtures. Moreover,  Algorithm \ref{alg:spectralclustering} is computationally more efficient since it runs clustering algorithms on an $n \times k$ matrix $\widehat{U}$, in contrast to the $n \times p$ matrix using the best rank-$k$ projections. 

\subsubsection{Stage 2: Feature Selection Using R-squared}
To select informative features, a first thought would be to compare the sum of squares $\sum_{i=1}^{n} Y_{ij}^2$ of different columns. The larger the sum of squares is, the more likely it is an informative feature. Indeed, when there is no signal, i.e. $j \in S^c$, the sum of squares is a sum of independent Chi-square random features with expectation $n\sigma_j^2$. And when there is a signal, the expectation of the sum of squares is $\sum_{a=1}^{k} n_a^* B_{aj}^2 +  n\sigma_j^2$. If $\sigma_j$'s are the same for all $j$, one would expect this method to correctly select informative features. However, $\sigma_j$ may vary in practice and we could have some $j_1$ and $j_2$ such that $\sum_{a=1}^{k} n_a^* B_{aj_1}^2 +  n\sigma_{j_1}^2 \ll n\sigma_{j_2}^2$. To avoid this problem, we need to normalize by the variance of each column. 
\begin{algorithm}
\caption{Feature Selection Using $R^2$} \vspace{0.05in}
\textbf{Input}: $Y_1, Y_2, \cdots, Y_n$. The number of clusters $k$. Initial estimates of clusters $G_1, G_2, \cdots, G_k$\vspace{0.05in}. A threshold $\tau \in (0,1)$. \\
\textbf{Output}: An index set $\widehat{S}$. 
\begin{itemize}
\item[1.] For $j=1,2,\cdots,p$, calculate:
\begin{itemize} 
\item[1a.] Estimated centers: $\hat{B}_{aj} = \frac{1}{|G_a|} \sum_{i \in G_a} Y_{ij}$
\item[1b.] Residual sum of squares:  $ c_j =  \sum_{a=1}^{k}\sum_{i \in G_a} (Y_{ij} - \hat{B}_{aj})^2 $ 
\item[1c.] Total sum of squares: $ m_j = \sum_{i \in [n]} (Y_{ij} - \bar{Y}_j)^2$
\item[1d.] Score: $SC_j = \frac{c_j}{m_j}$.
\end{itemize}
\item[2.] Output $\widehat{S} = \{j \in [n], SC_j \le \tau\}.$ 
\end{itemize}
\label{alg:select}
\end{algorithm}

To motivate our feature selection procedure, we consider a special case of symmetric, two balanced clusters with means $\theta$ and $-\theta \in \mathbb{R}^p$. Let $T_i \in \{1,2\}$ be the true label of $i$th sample, whose mean is $(2T_i-3)\theta$. For a non-informative feature $j \in S^c$, $\theta_j=0$. Thus ${\Var}(Y_{ij}|T_i) = {\Var}(Y_{ij})$. For informative feature $j\in S$, $\theta_j\neq 0$. For an informative feature, on the other hand, we have ${\Var}(Y_{ij}|T_i) < {\Var}(Y_{ij})$ for $j \in S$. Let $\widetilde{T}_i\in\{1,2\}$ be the cluster label for the $i$th sample obtained from Stage 1, it is natural to consider the quantity
\[ R^2_j= 1-\frac{\E[\Var(Y_{ij}|\widetilde{T}_i)]}{\Var(Y_{ij})}. \]

\begin{proposition}\label{prop:r2}
For $i$th example, let $a_{kl}  = \P(T_i=k,\widetilde{T}_i=l)$ for $k,l \in\{1,2\}$.
\begin{equation}\label{eq:r2_express}
R_j^2 = \frac{\theta_j^2}{\theta_j^2+\sigma_j^2} \left(\frac{(a_{11}-a_{21})^2}{(a_{11}+a_{21})}+ \frac{(a_{22}-a_{12})^2}{(a_{22}+a_{12})} \right).
\end{equation}
\end{proposition}

In the case of pure initial random guess $a_{11}=a_{21}$ and $a_{22}=a_{12}$, $R_j^2=0$. If the initial estimator $\widetilde{T}$ is slightly better than random guess, we have $R^2_j >0$ for informative feature. We can distinguish between $j \in S$ and $j \in S^c$ via $R_j^2$. Besides, when $j \in S$, $R^2_j$ depends on signal-to-noise ratio $\theta_j^2/\sigma_j^2$. The higher the signal-to-noise ratio, the weaker condition we need on the initial estimator to get the same $R^2_j$. We defer to Section \ref{sec:selectthm} for our detailed analysis on the sample version and the general number of clusters.

\subsubsection{Stage 3: Spectral Clustering and Lloyd's Algorithm}
With the features selected in Stage 2, the problem is reduced to low-dimensional Gaussian mixtures, which has been studied extensively in the literature. Among them, the most popular algorithms for Gaussian mixtures are the Lloyd's algorithm \cite{lloyd1982least}, EM algorithm\cite{dempster1977maximum}, methods of moments \cite{lindsay1993multivariate}, and tensor decompositions \cite{anandkumar2012method}. For stage 3, we use the spectral clustering Algorithm \ref{alg:spectralclustering} on selected features, followed by the Lloyd's iterations. The Lloyd's algorithm, often be referred as $k$-means algorithm, enjoys good statistical and computational guarantees for Gaussian mixture models \cite{lu2016statistical}. Given an initial estimator of the labels or centers, it iteratively updates the labels and centers on the selected features until convergence. A precise description is given in Algorithm \ref{alg:lloyd}. We refer the readers to \cite{lu2016statistical} for more discussions of the Lloyd's algorithm.

\begin{algorithm}
\caption{SpecLloyd algorithm} \vspace{0.05in}
\textbf{Input}: $Y_1, Y_2, \cdots, Y_n$. The number of clusters $k$. An index set of selected features $\widehat{S}$ \vspace{0.05in}. \\
\textbf{Output}:  Estimated cluster labels $z_1^{(T)}, z_2^{(T)}, \cdots, z_n^{(T)}$.  \vspace{-0.05in}\\
\begin{itemize}
\item[1. ] Run Algorithm \ref{alg:spectralclustering} on $\{\tilde{Y}_i\}$ to get an initial estimate of labels, $z_1^{(0)}, z_2^{(0)}, \cdots, z_n^{(0)}$, where $\tilde{Y}_i$ is the sub-vector of $Y_i$ with support $\widehat{S}$.
\item[2. ] Run the following iterations for $t=1,2,\cdots, T$.  \vspace{-0.05in}
\begin{itemize}
\item[2a.] For $(a, j) \in [k] \times \widehat{S}$,
\[ \widehat{B}_{aj}^{(t)} = \frac{\sum_{i \in [n]} Y_{ij} \mathbf{1}\{z_i^{(t-1)}=a\} }{\sum_{i \in [n]}  \mathbf{1}\{z_i^{(t-1)}=a\}}.  \]
\item[2b.] For $i \in [n]$,
\[ z_i^{(t)} = \argmin_{a \in [k]} \sum_{j \in \widehat{S}} (Y_{ij} - \widehat{B}_{aj}^{(t)})^2. \] 
\end{itemize}
\end{itemize}
\label{alg:lloyd}
\end{algorithm}

In summary, we first conduct spectral clustering to estimate noisy cluster labels, then we apply $R^2$ to select top informative features, finally we apply spectral clustering again on selected features. To further reduce the error, we apply the Lloyd's algorithm after the last stage.

\section{Convergence Analysis}
\label{sec:theo}
To better present our theoretical results, let us first introduce some notations and assumptions. For any partition $G$, we define a group-wise mislabeling rate. Recall that $T$ is the true partition. Let
$$B(G,T) = \min_{\pi \in \mathbb{S}_k} \max_{a \in [k]} \left\{ \frac{|G_{\pi(a)} \cap T_a^c|}{|G_{\pi(a)}|}, \frac{|T_a \cap G_{\pi(a)}^c|}{|T_a|} \right\},$$
where $\mathbb{S}_k$ is the set of permutations from $[k]$ to $[k]$. The two terms can be interpreted as the false positive rate and true negative rate of each group, respectively. 

Let $\alpha n = \min_{a \in [k]} n_a^*$ be the smallest cluster size, where $n_a^* = |T_a|$.  Since there are $k$ clusters, we have $\alpha$ strictly greater than $0$. $\alpha$ will play a role in our analysis because it determines how well we can estimate the centers even under the oracle case that the true labels are available. And it will further affect the quality of feature selections. 

Another crucial quantity in our analysis is the signal-to-noise ratio. We define  
\[ \textrm{SNR} = \min_{j \in S} \frac{1}{n\sigma_j^2}\sum_{a=1}^{k} n_a^* (B_{aj} - \bar{B}_{*j}) ^2\]
as the average signal-to-noise ratio of informative features, where $\bar{B}_{*j} = \frac{1}{n} \sum_{i=1}^n B_{ij}$. Intuitively, the larger $\textrm{SNR}$ is, the easier the clustering task is. In order to do non-trivial clustering, a necessary condition is that the signal strength is bigger than the noise level. Thus, we need a lower bound on \textrm{SNR}.

In the following, we split the convergence analysis into three parts, corresponding to the three stages of our algorithm. 

\subsection{Error Rate of Spectral Clustering} \label{sec:specthm}
The following theorem provides an upper bound on group-wise mis-clustering error of spectral clustering algorithm \ref{alg:spectralclustering} for Gaussian mixture model. 
\begin{thm} \label{thm:spec}
Let $G$ be the partition returned by Algorithm \ref{alg:spectralclustering} and $B_S$ be the sub-matrix of $B$ consist of $s$ non-zero columns.  Assume the $k$th singular value
\begin{equation} \label{eq:eigencond}
\sigma_k(B_S) \ge C \max \left\{ \sigma \sqrt{\frac{k}{\alpha}}, \left( \frac{\sigma^2 k p }{\alpha^2 n} \right)^{1/4} \right\}
\end{equation}
for a sufficiently large constant $C$. Then the group-wise mis-clustering error rate 
\[ B(G,T) \le \frac{C_1\sigma^2 k (\alpha n \sigma_k^2(B_S) + p)}{\alpha^3 n \sigma_k^4(B_S)} \]
with probability greater than $1-\exp(-C_2n)$ for some universal constants $C_1$ and $C_2$.  
\end{thm}
It guarantees a relatively small mis-clustering error, for example, $10\%$, under condition (\ref{eq:eigencond}). It only has a $(p/n)^{1/4}$ dependence on the dimensionality of the problem in condition (\ref{eq:eigencond}). Thus it is applicable to the high dimensional problem and can be satisfied under many interesting cases. For example, when $B_S$ is a random matrix, its minimum eigenvalue can be lower bounded by $c\sqrt{s}$ for some constant $c$ with high probability \cite{vershynin2010introduction}, where $s$ is the number of informative features. Then condition (\ref{eq:eigencond}) is reduced to $s \gtrsim \max\{\sigma^2, \sigma (p/n)^{1/4} \}$ by regarding $k$ and $\alpha$ as constants. 

As discussed in Section \ref{sec:spec}, another version of the spectral clustering algorithm is to run a distance-based clustering algorithm on the rows of $\widehat{Y}$,  the rank-$k$ approximation of the data matrix $Y$, instead of on the estimated eigenspace $\widehat{U}$. The condition \cite{awasthi2012improved, lu2016statistical} we need for this spectral clustering algorithm is
\[ \min_{u \neq v \in [k]^2} \|B_{u*} - B_{v*}\| \ge C_4 \sigma \sqrt{\frac{k}{\alpha}\left(1+\frac{kp}{n}\right)} \] 
for some sufficiently large constant $C_4$, since there are only $s$ non-zero entries of each row of $B$. It requires $s \gtrsim \sigma \sqrt{p/n}$ when $k$ and $\alpha$ are constants. Thus, Algorithm \ref{alg:spectralclustering} works better for the high-dimensional setting. 

\subsection{Feature Selection Guarantees} \label{sec:selectthm}
The next theorem provides theoretical guarantees of the feature selection step. 
\begin{thm} \label{lm:key1}
Assume $\textrm{SNR} > C_0$ for some sufficiently large constant $C_0$. Then there exist a constant $c$ such that for any given estimated partition $G$ (could be data dependent) with $B(G, T) \le c \alpha$. 
\begin{itemize}
\item[(a).] When $j \in S$, we have $SC_j \le 0.9$ with probability greater than $1-\exp(-cn)$. 
\item[(b).] When $j \in S^c$, we have $SC_j > 0.9$ with probability greater than $1-\exp(-c\alpha n)$
\end{itemize}
Therefore, when $\alpha n = \Omega(\log p)$, a choice of $\tau=0.9$ successfully selects all the informative features with probability greater than $1-\exp(-c\alpha n)$.
\end{thm}

Given any initializer with $B(G, T) \le c\alpha$, we are guaranteed to select all the informative features with high probability when $\alpha n = \Omega({\log p})$. It implies that the number of features $p$ is allowed to grow exponentially fast of the sample size $n$. Such scaling also appears in the feature selection problem under sparse linear regression model \cite{wainwright2009sharp}. Since feature selection only depends on the error rate of initial guess, we can also choose other clustering approaches in Stage 1 as long as the error rate is satisfactory.

\subsection{Error rate of the Lloyd's algorithm}
Finally, we have the following result from \cite{lu2016statistical} to characterize the performance of the Lloyd's algorithm. \begin{thm} \label{thm:lloyd}
Let $\Delta=\min_{u \neq v \in [k]^2} \|B_{u*}- B_{v*}\|$.  Assume $n \alpha^2 \ge Ck \log n$, $n \ge ks$ and $\Delta \ge C\sigma_S \sqrt{k/\alpha}$ for a sufficiently large constant $C$. Given any initializer $G_0$ satisfying 
\begin{equation} \label{eq:condlloyd}
B(G_0, T) < \frac{\min_{u \neq v \in [k]^2} \|B_{u*}- B_{v*}\|}{4\max_{u \neq v \in [k]^2} \|B_{u*}- B_{v*}\|}:= \frac{1}{4\lambda}
\end{equation} 
with probability $1-\nu$. Then
\begin{equation} 
\frac{1}{n} \sum_{i=1}^{n} \mathbf{1}\{\hat{z}_i^{(s)} \neq z_i \} \le \exp \left( - \frac{\Delta^2}{16\sigma_S^2}\right),
\end{equation}
for all $s \ge 4 \log n$ with probability greater than $1-\nu-4/n-2\exp(-\Delta/\sigma_S)$.
\end{thm}
Theorem \ref{thm:lloyd} states that we can achieve an exponentially small mis-clustering error after $\lceil 4\log n \rceil$ Lloyd's iterations given any initializer that satisfies condition (\ref{eq:condlloyd}). Suppose we have selected all the informative features in stage 2. By applying Theorem \ref{thm:spec} on the sub-matrix $B_S$, we obtain
\[ B(G_0, T) \le \frac{C_1 \sigma_S^2 k}{\alpha^2 \sigma_k^2(B_S)} \le \frac{1}{4\lambda} \]
when 
$\sigma_k(B_S) \ge C_2 \sigma_S \sqrt{\lambda k/\alpha^2}$
for some sufficiently large constant $C_2$. 

Combining the results of Theorem \ref{thm:spec}, Theorem \ref{lm:key1} and Theorem \ref{thm:lloyd}, we are able to give theoretical guarantees of our SC-FS algorithm. Let $\hat{z}=\{\hat{z}_1, \cdots, \hat{z}_n\}$ be the estimated labels returned by running SC-FS algorithm with $\tau=0.9$ and $T=\lceil 4\log n \rceil$. The following result upper bounds the mis-clustering error rate of $\hat{z}$.

\begin{thm} \label{thm:main}
Assume $\alpha n \ge C(\log p + k\log n/\alpha + \alpha ks)$, $\textrm{SNR} \ge C$, $\Delta \ge C \sigma_S \sqrt{k/\alpha}$ and 
\begin{equation} \label{eq:eigencond2}
\sigma_k(B_S) \ge \frac{C}{\alpha} \max \left\{ \sigma \sqrt{\frac{k}{\alpha}}, \sigma_S \sqrt{\lambda k},  \left( \frac{\sigma^2 k p }{\alpha^2 n} \right)^{1/4} \right\}
\end{equation}
for a sufficiently large constant $C$. Then
\begin{equation} 
\frac{1}{n} \sum_{i=1}^{n} \mathbf{1}\{\hat{z}_i \neq z_i \} \le \exp \left( - \frac{\Delta^2}{16\sigma_S^2}\right),
\end{equation}
with probability greater than $1-8/n-4\exp(-\Delta/\sigma_S)$.
\end{thm}

By Theorem 3.3 in \cite{lu2016statistical}, the minimax lower bound for clustering Gaussian mixture model is $\exp \left( - \frac{\Delta^2}{8\sigma_S^2}\right)$. The worst case constructed in \cite{lu2016statistical} can be naturally generalized to the sparse Gaussian mixture model. Therefore, the proposed SC-FS algorithm is rate-optimal up to a constant factor in the exponent. Note that the mis-clustering rate only takes value in $\{0, 1/n, 2/n, \cdots, 1\}$. Theorem \ref{thm:main} guarantees a perfect clustering when $\Delta > 4\sigma_S \log n $. 

\subsection{Tuning parameter selection}
\label{subsec:tuning}
\subsubsection{Number of clusters}
For each possible $k = 1,...,20$, we conduct the following steps:
\begin{enumerate}
\item
Conduct SVD on data matrix and obtain top $k$ left singular vectors as matrix $U\in \mathbb{R}^{n\times k}$. 
\item
Conduct $k$-means clustering algorithm of $U$.
\item
Calculate the ratio of within cluster sum of squares and total sum of squares as unexplained variation ratio $\eta(k)$. And let $\xi(k) = 1-\eta(k)$ be the variation explained ratio.
\end{enumerate}
We plot $\xi(k)$ versus $k$ and select the change point as the number of clusters. 

\subsubsection{Feature Selection Threshold}
\label{sec:sparsity}
The actual threshold depends on the error rate of initializer and on the quality $\alpha n / \log p$. As suggested by Theorem \ref{lm:key1}, we could use $\tau = 0.9$ as a practical guidance of the feature selection threshold.

\section{Numerical Experiments}
\label{sec:experiments}
\subsection{Synthetic data generation}
\label{exp: syn_data_gen}
Let $k$ be the number of clusters, $n$ be the number of samples, $p$ be the number of features, $s$ be the number of informative features, and $\sigma_{k}$ be the signal strength introduced in Theorem \ref{thm:spec}. 
For a set of ($k$, $n$, $p$, $s$, $\sigma_{k}$), we generate data as follows:
\begin{enumerate}
\item
Generate elements of $\tilde{B}\in \R^{k\times s}$ as left singular matrix of $i.i.d.$ $s\times s$ standard Gaussian random matrix. We get $B\in \R^{k\times p}$ as $B = [\sigma_k\tilde{B},\bd{0}_{k\times(p-s)}]$. 
\item
Generate the cluster label $z_i\in\{1,...,k\}$ of the $i$th sample by randomly assigning. Then generate membership matrix $Z\in\R^{n\times k}$ with $Z_{ij} = \mathbbm{1}{(j=z_i)}$.
\item
Generate data matrix $Y = ZB + W$, where $W$ is standard Gaussian noise matrix (or $t_2$ noise matrix if specified). Then we scale the columns of the data matrix.
\end{enumerate} 

\subsection{Convergence rate of spectral clustering}
In this simulation, we numerically evaluated the convergence rate of spectral clustering. To study the effect of the number of features $p$ on the error rate of spectral clustering, we fixed the number of clusters $k=4$, the number of observations $n=100$, the number of features $p=100$, the number of informative features $s=100$, and the signal strength $\sigma_k = 4$. We varied $p$ from $100$ to $1000$, $n$ from $100$ to $1000$, and $\sigma_k$ from $2$ to $5$ to study the error convergence rate regarding each factor ($n$, $p$, or $\sigma_k$) with two other factors fixed. For each setting of ($k$, $n$, $p$, $s$, $\sigma_k$), we generated synthetic data according to Section \ref{exp: syn_data_gen} with Gaussian noise and applied spectral clustering according to Algorithm \ref{alg:spectralclustering}. We repeated the above process for 50 times and computed the average error rate. The scatter plots are shown in Figure \ref{fig:err_rate_check}. We observe a linear relationship between error rate and $p$, and also expected rate for $n$ and $\sigma_k$. 

In terms of spectral clustering with sparse informative features, we can improve the clustering result to a great extent if the number of informative features $s$ is much smaller than the total number of features, given that we have selected all informative features. Even if we fail to select all informative features, we can still have a better clustering result as long as we have selected enough features such that signal-to-noise ratio does not decrease too much after feature selection.  

\begin{figure}[ht]
	\centering
	\subfigure{\includegraphics[width=50mm]{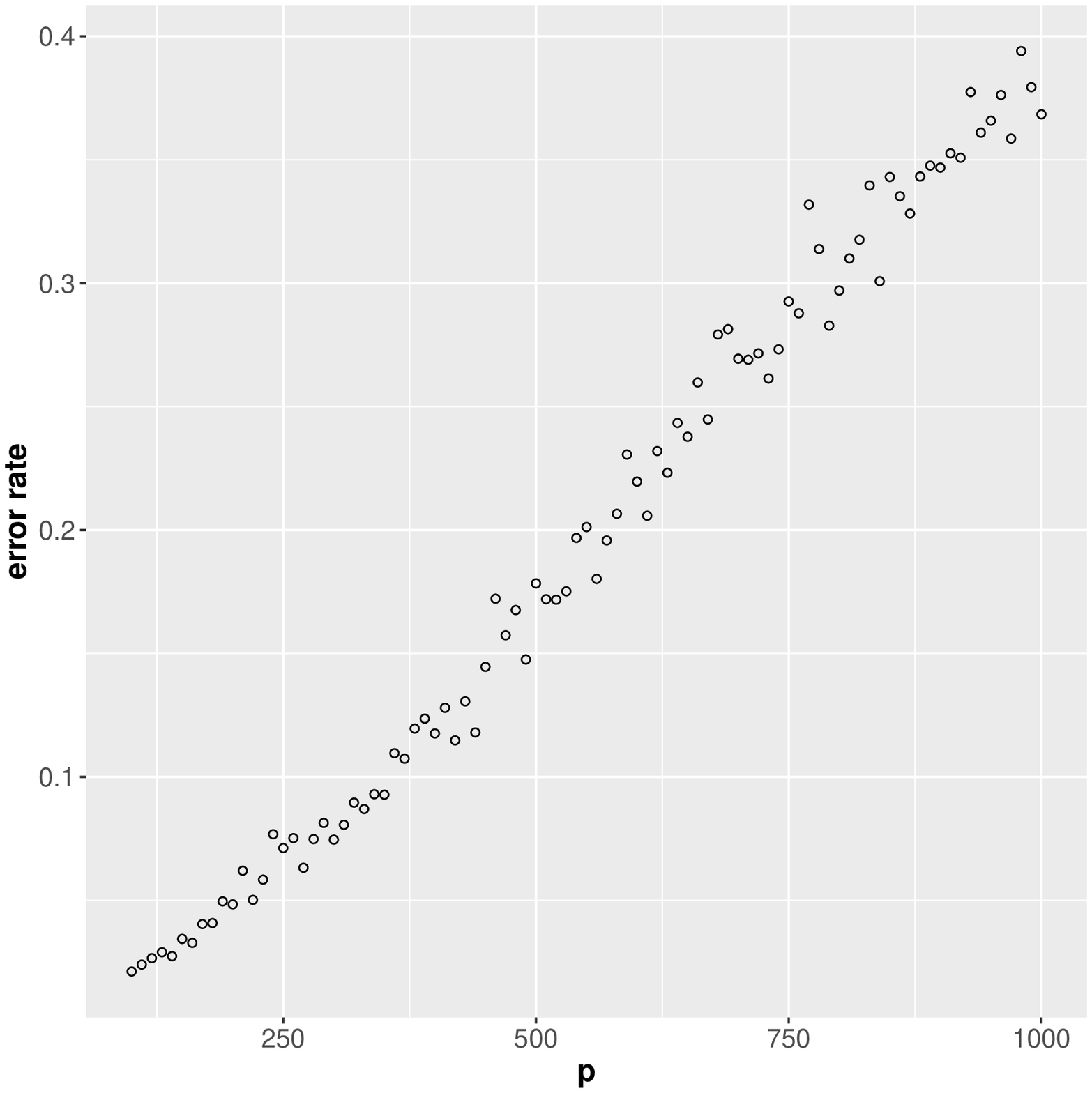}}
	~ \subfigure{\includegraphics[width=50mm]{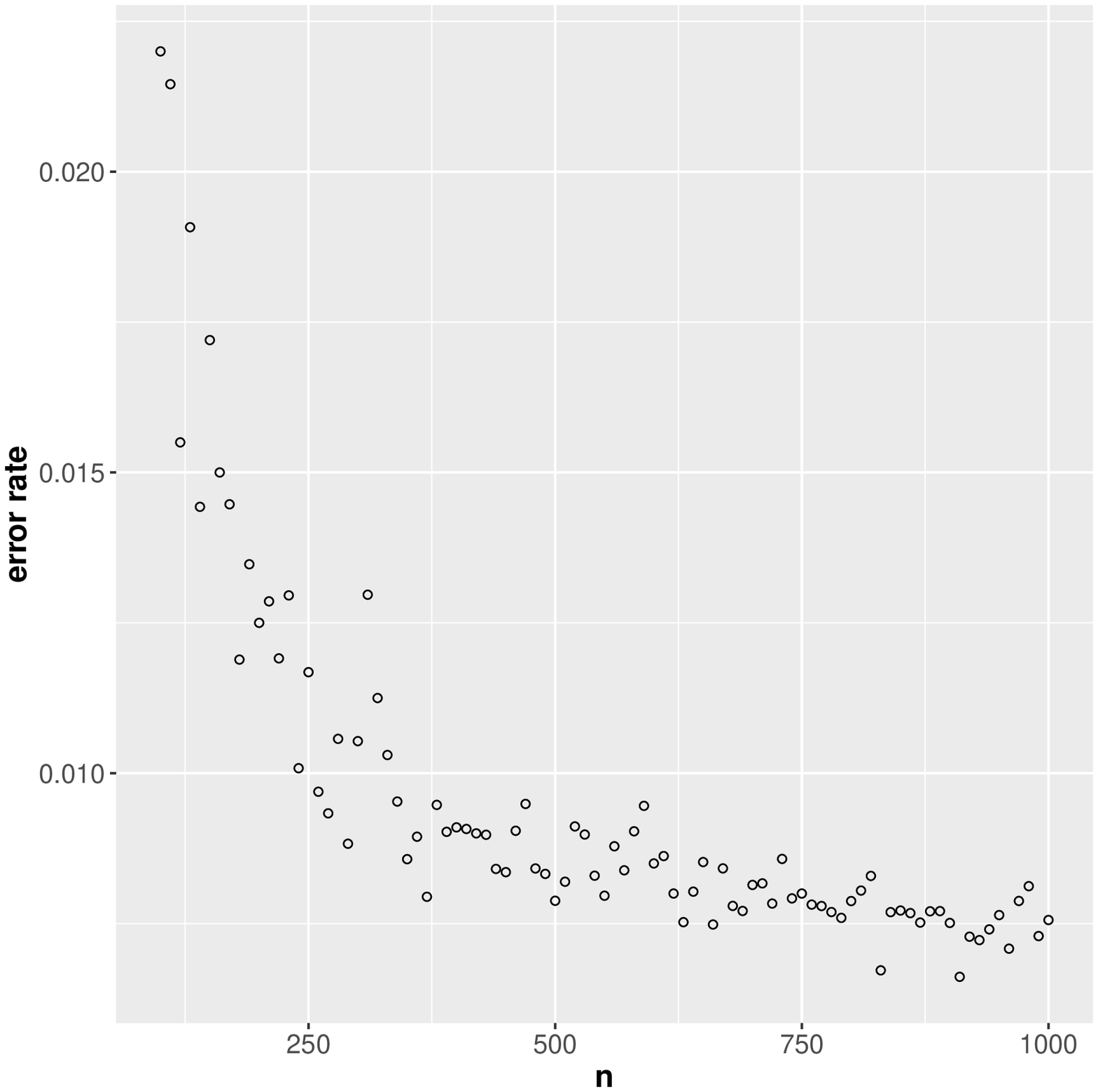}}
	~ \subfigure{\includegraphics[width=50mm]{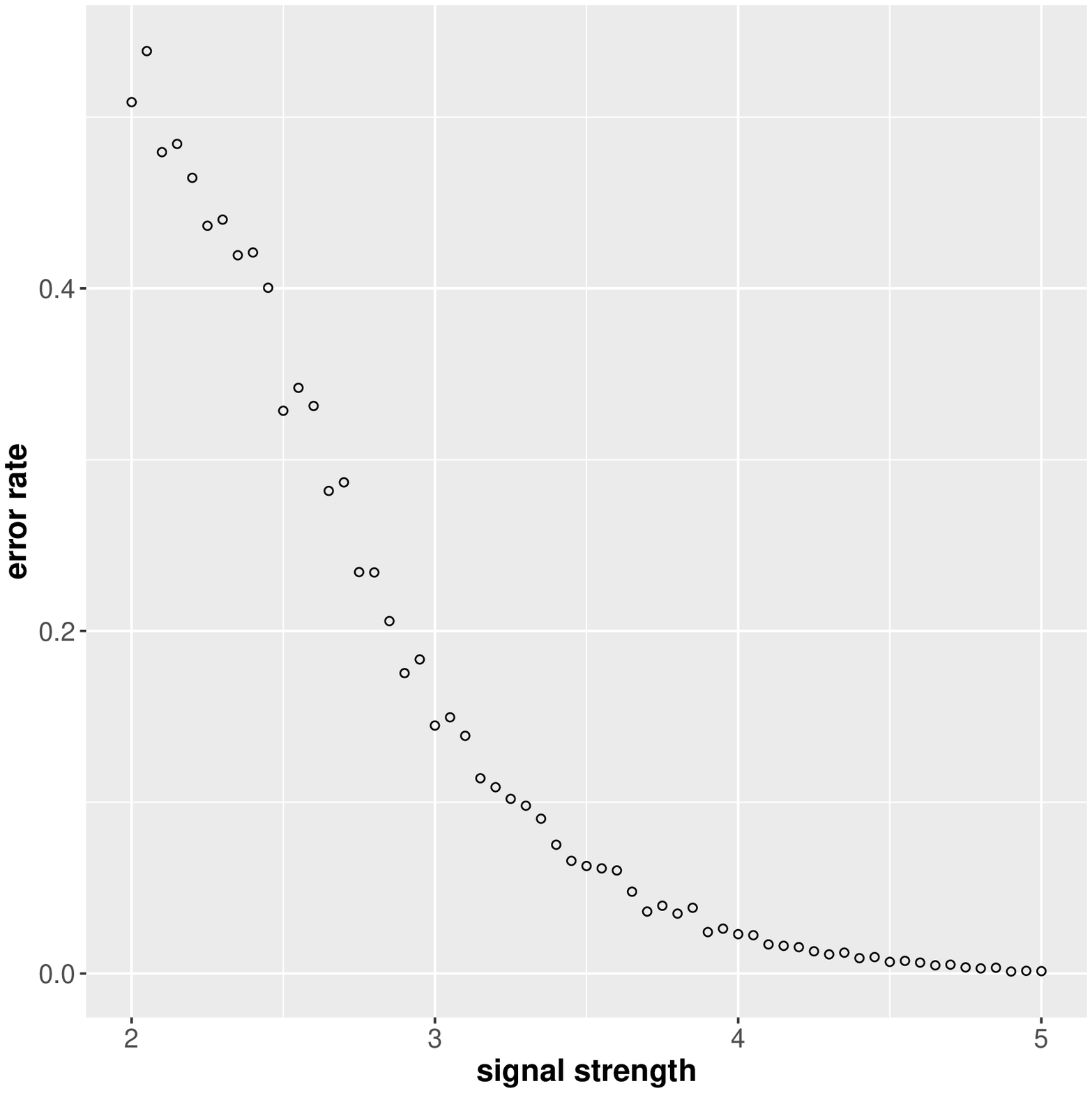}}
	\caption{Convergence rate of error rate}
	\label{fig:err_rate_check}
\end{figure}

\subsection{Feature selection $F_1$}
In this simulation, we studied the relationship between feature selection success metrics and quality of initial guess. We fixed $k=4$, $s=100$, $p=500$, and varied $n\in\{10,50,100\}\log p$.  Let the true label of the $i$th observation be $l_i$, and the initial guessed label be $\hat{l}_i$. We define the initial guess error rate as: 
$$
\mathcal{M}( l,\hat{l} ) =\frac{1}{n} \min_{\pi} {\left|\left\{i : l_i\neq\pi(\hat{l}_i)\right\}\right|}.
$$   
We create guessed labels with the given error rate taking values from $\{0.05,0.1,0.15,0.2,0.3\}$. We set $\sigma_k \in \{5,10\}$. 

Let $\mathcal{S}$ be the set of true informative features with $|\mathcal{S}| = s$, and $\hat{\mathcal{S}}$ be the set of estimated informative features based on $R^2$ Algorithm \ref{alg:select}. We compute the $F_1$ score to measure the feature selection quality. For a set of ($n$, $p$, $s$, $\sigma_k$), we generated data as described in Section \ref{exp: syn_data_gen} and repeated the experiment 50 times. Given the membership matrix $Z$, we generated the guessed label $\tilde{Z}$ equal to $Z$ with probability $1-\eta$, and equal to one of other $k-1$ values with equal probability $\eta/(k-1)$.

We can observe that as signal strength $\sigma_k$ increases, $n$ increases, and mis-clustering rate of initial guess decreases, the feature selection performance improves (Table \ref{tab:fs_init_guess}). When the signal strength and number of samples are large enough, the selected features are of high quality. This observation is consistent with Equation (\ref{eq:r2_express}) and Theorem \ref{lm:key1}.

\begin{table}[ht]
\centering
\caption{Feature selection $F_1$ scores averaged over 50 runs. Numbers in the brackets are the standard deviations.}
\label{tab:fs_init_guess}
\begin{tabular}{|c|c|c|c|c|c|c|}
\hline
&&\multicolumn{5}{|c|}{initial guess error rate}\\\hline
$\sigma_k$&$n/\log p$&0.05&0.1&0.15&0.2&0.3\\\hline
5  & 10  & 0.620 (0.032) & 0.604 (0.037) & 0.578 (0.044) & 0.540 (0.052) & 0.456 (0.071) \\\hline
5  & 50 & 0.744 (0.027) & 0.698 (0.042) & 0.626 (0.045) & 0.548 (0.049) & 0.311 (0.072) \\\hline
5  & 100 & 0.742 (0.034) & 0.671 (0.039) & 0.593 (0.04)  & 0.507 (0.045) & 0.256 (0.067) \\\hline
10 & 10  & 0.736 (0.027) & 0.731 (0.024) & 0.720 (0.027) & 0.701 (0.031) & 0.664 (0.040) \\\hline
10 & 50 & 0.958 (0.013) & 0.949 (0.014) & 0.935 (0.018) & 0.909 (0.025) & 0.821 (0.045) \\\hline
10 & 100 & 0.956 (0.015) & 0.948 (0.016) & 0.932 (0.019) & 0.908 (0.022) & 0.816 (0.032) \\\hline
\end{tabular}
\end{table}

\subsection{Comparisons on Synthetic Data}
\subsubsection{Gaussian noise}\label{subsubsec:gaus}
In this simulation, we fixed $k=4$, $p=8000$, $s=500$, $\sigma_k = 6$, and $n/\log p = 15,20,25,30$. we generated synthetic data according to Section \ref{exp: syn_data_gen}. We denote SC-FS1 as spectral clustering in stage 3, and SC-FS2 as Lloyd iteration following SC-FS1. We compared our methods SC-FS1 and SC-FS2 with spectral clustering, spectral plus Lloyd clustering (specLloyd, for short) \citep{lu2016statistical}, model-based clustering (mclust) \citep{scrucca2016mclust}, and sparse K-means (spKmeans, for short) \citep{witten2012framework}. As shown in Table \ref{tab:comp_normal}, our proposed methods performed the best and Lloyd iteration in stage 3 improved SC-FS1 to a small extent. By comparing specLloyd with the proposed method, we can observe that feature selection in stage 2 can reduce the error rate.

\begin{table}[ht]
\centering
\caption{Comparisons of different methods under Gaussian noise averaged over 50 runs. Numbers in the parenthesis are the standard deviations of the error rate.}
\label{tab:comp_normal}
\begin{tabular}{|c|c|c|c|c|c|}
\hline
$n/\log p$&specLloyd&mclust&spKmeans&SC-FS1&SC-FS2\\\hline
15&0.541(0.065)&0.606(0.034)&0.612(0.051)&0.539(0.076)&\textbf{0.524}(0.072)\\\hline
20&0.406(0.104)&0.626(0.046)&0.463(0.079)&0.392(0.068)&\textbf{0.391}(0.103)\\\hline
25&0.277(0.088)&0.561(0.050)&0.302(0.146)&0.208(0.097)&\textbf{0.202}(0.085)\\\hline
30&0.196(0.037)&0.601(0.040)&0.081(0.132)&0.054(0.026)&\textbf{0.053}(0.024)\\\hline
\end{tabular}
\end{table}

\subsubsection{Heavy-tailed noise}
In this simulation, we compare the methods in heavy-tailed noise case to study the robustness of the proposed method. We followed the same setting as in \ref{subsubsec:gaus} in generating the synthetic data, except that we used standard $t_2$ distribution to generate noise. The proposed approach shows advantage under the heavy-tailed noise case (Table \ref{tab:comp_t}), while spKmeans does not converge well with sample size growth. To some extend, this suggests that our proposed approach is robust to heavy-tailed noise.
\begin{table}[ht]
\centering
\caption{Comparisons of different methods under $t_2$ noise averaged over 50 runs. Numbers in the parenthesis are the standard deviations of the error rate.}
\label{tab:comp_t}
\begin{tabular}{|c|c|c|c|c|c|}
\hline
$n/\log p$&specLloyd&mclust&spKmeans&SC-FS1&SC-FS2\\\hline
15 & 0.478(0.082) & 0.612(0.076) & 0.673(0.029) & 0.516(0.111) & \textbf{0.468}(0.115) \\\hline
20 & 0.367(0.090) & 0.543(0.151) & 0.682(0.039) & 0.335(0.175) & \textbf{0.298}(0.159) \\\hline
25 & 0.256(0.078) & 0.461(0.213) & 0.705(0.019) & 0.189(0.136) & \textbf{0.175}(0.141) \\\hline
30 & 0.199(0.071) & 0.372(0.272) & 0.712(0.016) & 0.119(0.134) & \textbf{0.102}(0.102) \\\hline
\end{tabular}
\end{table}

\subsection{Real Data}
\label{sec:real}
\subsubsection{Dataset description}  
{We compared clustering results of our method with other methods on four publicly available high-dimensional datasets. We selected these datasets because they represent a wide range of high-dimensional data with different numbers of data points and classes from various fields. Characteristics of the four real datasets are summarized in Table \ref{tab:char_sc_data}.} The details of four datasets are as follows:

\begin{enumerate}
\item
{Zheng: The Peripheral blood mononuclear cells (PBMC) scRNA-seq data were generated by the 10x Genomics GemCode protocol. We obtained the data from the package DuoClustering2018 \citep{duo2019package} with ExperimentHub ID “EH1532”. The data consist of eight cell types in approximately equal proportions. We first performed library size normalization through dividing the counts by the total UMI in that cells, multiplying the resulting fraction by 10,000, and doing log transformation. Then, feature scaling is carried out using the function scale.}
\item
{Yeoh: The bone marrow microarray data were downloaded from R package datamicroarray \citep{ramey2016datamicroarray}. The 248 samples were obtained from pediatric acute lymphoblastic leukemia patients with six subtypes, including T-ALL, E2A-PBX1, TEL-AML1, BCR-ABL, MLL, and HK50. The number of features, i.e. genes, is 12,625.}
\item
% Usoskin: Neuronal cells with sensory subtypes \citep{usoskin2015unbiased}.

% TODO Tianqi
{BBC: This dataset has 2,225 articles with 1,490 for training and 735 for testing. Each article has one label from five categories: business, entertainment, politics, sports or tech. We downloaded the data from \cite{greene06icml} and used the training data to compare among different clustering algorithms. We did not use test data because there are no labels available from the dataset. The 1,490 articles with five categories were processed by term frequency–inverse document frequency (tf-idf) vectorizer. We obtained 24,746 features as a result.}
\item

% TODO Tianqi
{Agnews: This dataset is a collection of more than 1 million news articles. The AG's news topic classification dataset was constructed by choosing four largest classes from the original corpus. Each class contains 30,000 training samples and 1,900 testing samples. The total number of training samples is 120,000 and that of testing samples is 7,600. We downloaded the data from \cite{zhang2015character} and used the test set to compare different clustering algorithms. We used the test data because it has thousands of examples with tens of thousands of features (after tf-idf), which fits the high-dimensional setting. The 7,600 articles with four categories are also processed by tf-idf vectorizer. We obtained 21,853 features as a result.}

\end{enumerate}

% Cell types in each dataset were known and validated in respective studies, which provided us a gold standard assessing clustering performance.

\begin{table}[ht]
\centering
\caption{Summary of characteristics of the four real datasets}
\label{tab:char_sc_data}
\begin{tabular}{|c|c|c|c|}
\hline
% Dataset & \# cells & \# genes & \# cell types\\\hline
% Kolod&704&13473&3\\\hline
% Usoskin&622&17772&4\\\hline
% 10x&1889&58302&6\\\hline
% Zheng&2295&15716&5\\\hline
Dataset & \# data & \# features & \# classes\\\hline
Zheng&3994&15716&8\\\hline
Yeoh&248&12625&6\\\hline
BBC&1490&24746&5\\\hline
agnews&7600&21853&4\\\hline
\end{tabular}
\end{table}

\subsubsection{Numerical comparisons among different methods}

{We performed comparisons of SC-FS on the four datasets to test its performance with three other methods including spectral clustering \citep{rohe2011spectral}, sparse K-means \citep{witten2012framework}, and K-means \citep{macqueen1967some}. For sparse K-means, we subsampled 1,500 data points for Zheng, Yeoh, agnews to avoid run time and memory issues. The adjusted Rand index (ARI) is shown in Table \ref{tab:ari_real}. SC-FS2 performed the best on three out of four datasets, and SC-FS1 resulted in the highest ARI on the remaining dataset, followed by spectral clustering.}

\begin{table}[ht]
\centering
\caption{ARI on four real datasets}
\label{tab:ari_real}
\begin{tabular}{|c|c|c|c|c|c|c|c|}
\hline
% Dataset& SC-FS1 & SC-FS2 & spectral & spKmeans  & Seurat & FEAST\\\hline
% Kolod &\textbf{1}&0.996&0.976&0.964&0.857&\textbf{1}\\\hline
% Usoskin &\textbf{0.688}&0.553&0.604&0.205&0.478&0.039\\\hline
% 10x &\textbf{0.182}&0.155&0.164&0.152&0.152&0.104\\\hline
% Zheng &0.365&\textbf{0.501}&0.345&0.498&0.472&0.086\\\hline
Dataset& SC-FS1 & SC-FS2 & spectral & spKmeans  & Kmeans\\\hline
Zheng &0.431&\textbf{0.437}&0.330&0.418&0.319\\\hline
Yeoh &\textbf{0.647}&0.579&0.554&0.337&0.258\\\hline
BBC &0.647&\textbf{0.658}&0.647&0.0440&0.573\\\hline
agnews &0.192&\textbf{0.205}&0.201&0.0151&0.180\\\hline
\end{tabular}
\end{table}

% \subsubsection{Selecting the number of clusters}
% As described in Section \ref{subsec:tuning}, we plot $\xi(k)$ versus $k$ (Figure \ref{fig:chooseK_real}). From the plots, we can observe the proposed approach can identify the correct number of clusters in all datasets except 10x. We hypothesize that it is because two clusters are very hard to distinguish from the linear approach without considering higher-order or non-linear effects. This can also be verified in Figure \ref{fig:umap}.
 
% \begin{figure}[ht]
% 	\centering
% 	\includegraphics[width=150mm]{figures/choosek_final_v2.eps}
% 	\label{fig:chooseK_real}
% 	\caption{Variation explained versus number of clusters in four real dataset}
% \end{figure}

% \subsubsection{Visualization before and after feature selection}
% To study the effect of the feature selection step, we visualize the RNA-seq Data using UMAP \citep{mcinnes2018umap} before and after the feature selection. Before our feature selection, Kolod results in three large clusters, and for Usoskin cells from the same cell type are loosed grouped but all connected. As for 10x and Zheng, UMAP only returns one blob, with all cells being mixed together. Whereas after SC-FS feature selection, all the datasets show improved visualization with clearer clusters. Kolod's three clusters become much more distinct, and most cell types can be separated for 10x and Zheng.
% \begin{figure}[ht]
% 	\centering
% 	\subfigure{\includegraphics[width=150mm]{figures/compare_umap_v1.eps}}
% 	\caption{UMAP visualization before and after feature selection}
% 	\label{fig:umap}
% \end{figure}

\section{Conclusions}
\label{sec:dis}
In this article, we proposed a three-stage algorithm that is minimax optimal for estimating the underlying cluster labels under the generative model of sparse Gaussian mixture model (\ref{eq:model}). Our method is able to identify all informative features given any initial estimator with $o(1)$ clustering error and theoretically verified the optimality of proposed method under sparse Gaussian mixture assumptions. We further demonstrated the power of the methods via extensive simulation studies and real data analysis. For further directions, it is interesting to explore the performance of our algorithm under other generative models with heavy tails. Based on the proposed framework, it is also interesting to compare other clustering and feature selection methods including nonlinear methods such as kernel methods and neural networks.  

\section*{Data Availability Statement}
The data that support the findings in this paper are openly available in Kaggle BBC (Broadcasting company) News Classification at \url{https://www.kaggle.com/c/learn-ai-bbc}, and AG News at \url{https://github.com/mhjabreel/CharCnn_Keras/tree/master/data/ag_news_csv}.

\bibliographystyle{biom}
\bibliography{scafs.bib}

\section*{Supporting Information}
Web Appendices, Tables, and Figures referenced in Sections \ref{sec:experiments} are available with this paper at the Biometrics website on Wiley Online Library. The code is available both on the Biometrics website and at \url{https://github.com/TerenceLiu4444/SCFS}.

\section{Proofs}
\label{sec:proof}
\subsection{Proof of Proposition 2.1}
For brevity, we denote $Y_{ij}$ as $X_j$, $T_i$ as $T$, and $\widetilde{T}_i$ as $\widetilde{T}$. For $j\in S$, by the decomposition of variance, we have
\begin{align*}
\Var(X_{j}|\tilde{T}=1) &= \E\Var(X_j|\widetilde{T}=1,T) + \Var ( \E[X_j|\widetilde{T}=1,T]) \\
				& = \sigma_j^2 + \Var \left ( (2T-3)\theta_j|\widetilde{T}=1 \right)
				\\
				& = \sigma_j^2+\theta_j^2 \left[ 1 -  \left( \frac{a_{11}-a_{21}}{a_{11}+a_{21}} \right)^2\right]
\end{align*}
For the above last equality, it is because
$\E[(2T-3)^2|\widetilde{T}=1] = 1$
and $\E[(2T-3)|\widetilde{T}=1] = \frac{a_{11}-a_{21}}{a_{11}+a_{21}}$.

Similarly
$$
\Var(X_{j}|\widetilde{T}=2)= \sigma_j^2+\theta_j^2 \left[ 1 -  \left( \frac{a_{22}-a_{12}}{a_{22}+a_{12}} \right)^2\right]
$$

Notice $\P(\widetilde{T}=1) = a_{11} + a_{21}$, $\P(\widetilde{T}=2) = a_{12} + a_{22}$, and $a_{11} + a_{21} + a_{12} + a_{22} = 1$, we have
\begin{align*}
\E \left( \Var(X_{j}|\widetilde{T}) \right)=& \ \P(\widetilde{T}=1)\left\{\sigma_j^2+\theta_j^2 \left[ 1 -  \left( \frac{a_{11}-a_{21}}{a_{11}+a_{21}} \right)^2\right]\right\} \\
& + \P(\widetilde{T}=2)\left\{\sigma_j^2+\theta_j^2 \left[ 1 -  \left( \frac{a_{22}-a_{12}}{a_{22}+a_{12}} \right)^2\right]\right\} \\
=& \sigma_j^2 + \theta_j^2\left[ 1-\left( \frac{(a_{11}-a_{21})^2}{(a_{11}+a_{21})}+ \frac{(a_{22}-a_{12})^2}{(a_{22}+a_{12})} \right) \right].
\end{align*}
On the other hand,
$$
\Var(X_j)=\E\Var(X_j|T) + \Var(\E[X_j|T]) = \sigma_j^2+\theta_j^2.
$$ 
Then,
\begin{equation}
R_j^2 = \frac{\theta_j^2}{\theta_j^2+\sigma_j^2} \left(\frac{(a_{11}-a_{21})^2}{a_{11}+a_{21}}+ \frac{(a_{22}-a_{12})^2}{a_{22}+a_{12}} \right).
\end{equation}
\subsection{Proof of Theorem 3.1} \label{sec:specproof}
The main proof idea of Theorem 3.1 follows from \citep{lei2013consistency}. Its proof is modular, which is based on two existing results in the literature. First, we need a perturbation bound on the eigenspaces. The traditional Wedin's sin $\Theta$ Theorem gives the same perturbation bound for the left and right singular subspaces, which is sub-optimal under our setting. To capture the high-dimensional structure ($p \gg n$), we utilize the results in \citep{cai2016rate}. 

\begin{lemma} \label{lm:sintheta}
Suppose $X \in \mathbb{R}^{n \times p}$ is a rank $k$ matrix and $Z \in \mathbb{R}^{n \times p}$ whose entries are independent sub-gaussian random variables satisfying $\mathbb{E} e^{tZ_{ij}} \le e^{t^2 \sigma^2/2}$ for any $t>0$. Let $U$ be the left singular vectors of $X$ and $\widehat{U}$ be the top $k$ leading left singular vectors of $Y = X+Z$. Then there exist constants $C_1$ and $C_2$ such that 
\[ \inf_{O \in \mathbb{O}_k} \| \widehat{U} - UO \|_F \le \frac{C_1 \sigma \sqrt{kn (\sigma_k^2(X)+p)}}{\sigma_k^2(X)} \]
with probability greater than $1-\exp(-C_2n)$. Here $\mathbb{O}_k = \{ A \in \mathbb{R}^{k \times k}, A^T A = \mathbf{I}_k \}$ is the set of $k$-dimensional orthogonal matrices. 
\end{lemma}

Another key ingredient of our proof is the error bound for approximate k-means from \citep{lei2013consistency}.
\begin{lemma} \label{lm:appkmeans}
For $\epsilon>0$ and any two matrices $\widehat{U}, U \in \mathbb{R}^{n \times k}$ such that $U = ZQ$ with $Z \in \mathcal{Z}$ and $Q \in \mathbb{R}^{k \times k}$, let $\widehat{Z}, \widehat{Q}$ be a $(1+\epsilon)$-approximate solution to the $k$-means problem in equation (2) from the paper and $\widetilde{U} = \widehat{Z} \widehat{Q}$. For any $\delta_a \le \min_{b \neq a} \|Q_{b*} - Q_{a*}\|$, define $S_a = \{b \in T_a,  \|\widetilde{U}_{b*} - U_{b*}\| \ge \delta_a/2 \}$, then
\begin{equation}
\sum_{a=1}^{k} |S_a| \delta_a^2 \le 4(4+2\epsilon) \|\widehat{U} - U\|_F^2.
\end{equation}
Moreover, if
\begin{equation} \label{eq:cond1}
(16+8\epsilon) \| \widehat{U} - U \|_F^2 < n_a^* \delta_a^2 \qquad \textrm{for all } a \in [k],
\end{equation}
then there exists a permutation matrix $J \in \mathbb{R}^{k \times k}$ such that $(\widehat{Z}J)_{i*}= Z_{i*}$ for all $i \in \bigcup_{a=1}^{k} (T_a \backslash S_a)$.
\end{lemma}

Now we are ready to prove Theorem 3.1. In the following, we use a generic notation $C$ to denote absolute constants, whose value may vary from context to context.  By the above Lemma 1 and Lemma 1 from the paper, there exists an orthogonal matrix $O \in \mathbb{R}^{k \times k}$ such that 
\begin{equation} \label{eq:uhat_err}
\|\widehat{U} - UO\|_F^2 \le \frac{C\sigma^2 kn(\sigma_k^2(X)+p)}{\sigma_k^4(X)} \le \frac{C \sigma^2 k (\alpha n \sigma_k^2(B) + p)}{\alpha^2 n \sigma_{k}^4(B)},  
\end{equation}
with probability greater than $1-\exp(-Cn)$. For $UO = ZQO := Z\widetilde{Q}$,  Lemma 1 from the paper implies
\[ \|\widetilde{Q}_{b*} -  \widetilde{Q}_{a*}\| = \|Q_{b*} -  Q_{a*}\| \ge \frac{1}{\sqrt{n_a^*}}  \]
for all $b \neq a$. Applying Lemma \ref{lm:appkmeans} to $\widehat{U}$ and $UO$ with $\delta_a = 1/\sqrt{n_a^*}$, we obtain
\[ \sum_{a=1}^{k} \frac{|S_a|}{n_a^*} \le C \|\widehat{U} - UO\|^2_F. \] 
Let $\mathcal{E}$ be the event that (\ref{eq:uhat_err}) holds. Then on event $\mathcal{E}$, 
$$\max_{a \in [k]} \left\{ \frac{|S_a|}{n_a^*} \right\} \le \frac{C\sigma^2 k (\alpha n \sigma_k^2(B) + p)}{\alpha^2 n \sigma_k^2(B)} \triangleq R.$$
When $\sigma_k(B) \ge C \max \left\{ \sigma \sqrt{\frac{k}{\alpha}}, \left( \frac{\sigma^2 k p }{\alpha^2 n} \right)^{1/4} \right\}$ for a sufficiently large constant $C$, condition (\ref{eq:cond1}) satisfies. Without loss of generality, we assume the permutation matrix in Lemma \ref{lm:appkmeans} is identity matrix. Consequently, 
$|T_a \cap G_a^c| \le |S_a| \le n_a^* R$ for all $a \in [k]$. Note that $G_a \cap T_a^c \subseteq \bigcup_{b \in [k]} (T_b \cap G_b^c)$. We have  $|G_a \cap T_a^c | \le \sum_{b \in [k]} |T_b \cap G_b^c| \le n R,$
which implies 
\[ \frac{|G_a \cap T_a^c|}{|G_a|} \le \frac{nR}{|G_a \cap T_a|} \le \frac{nR}{n_a^*(1-R)} \le \frac{2}{\alpha} R \]
for all $a \in [k]$. Here the last inequality is due to the condition. Therefore, the desired result holds on event $E$.

\subsection{Proof of Theorem 3.2}
Let us first introduce some notations. Let $T_a \subseteq [n]$ be the true clusters. $G_a \subseteq [n]$ be the estimated clusters with cardinality $n_a$. For any $a \in [k]$, define $U_{a} =  \sum_{i \in G_a} w_i$ and $V_a= \sum_{i \in G_a} w_i^2$. For any sequence $b$, define $\bar{b}_a = \frac{1}{n_a} \sum_{i \in G_a} b_i$ and $\bar{b} = \frac{1}{n} \sum_{i=1}^{n} w_i$. With a little abuse of notation, we also define $\bar{\theta}_a = \frac{1}{n_a} \sum_{i \in G_a} \theta_{z_i}$. The analyses below are for a fixed $j$ and we denote by $x_i = Y_{ij}$, $\theta_a = B_{aj}$ and $w_{i} = W_{ij}$.  We also need the following two lemmas on the concentration behavior of $w_i$.

\begin{lemma} \label{lm:tech1}
There is a constant $c$ such that the following holds with probability greater than $1-\exp(-c n)$,
\begin{equation} \label{eq:tech1.1}
\left| \sum_{i=1}^{n} w_i^2 - n \sigma^2 \right| \le 0.1 n \sigma^2,
\end{equation}
\end{lemma}
\begin{proof}[Proof of Lemma \ref{lm:tech1}]
Note that $\sum_{i=1}^{n} w_i^2$ are sub-exponential random variables with expectation $n\sigma^2$. Bernstein equality gives us the desired result. 
\end{proof}

\begin{lemma} \label{lm:tech3}
Let $a_1, a_2, \cdots, a_n$ and $b_1, b_2, \cdots, b_m$ be two sequencesF of real numbers. Then $\sum_{i=1}^{m} (a_i - \bar{a})^2 = \sum_{i=1}^{m} a_i^2 - m \bar{a}^2$ and
\[ \sum_{i=1}^{m} a_i^2 \sum_{i=1}^{m} b_i^2 - \left( \sum_{i=1}^{m} a_i b_i\right)^2  = \frac{1}{2} \sum_{1\le i,j \le m} (a_i b_j - a_j b_i)^2 \]
\end{lemma}

\begin{proof}[Proof of Part (a)] Now we are ready to analyze the score for variable $j$. Let us first upper bound the conditional variance $c_j$. Using the fact that $(u+v)^2 \le 2u^2 + 2v^2$, we have
\begin{equation} \label{eq:condvar}
c_j = \sum_{a=1}^{k} \sum_{i \in G_a} (x_i - \bar{x}_a)^2 \le 2 \sum_{a=1}^{k} \sum_{i \in G_a} (\theta_{z_i} - \bar{\theta}_a)^2 + 2 \sum_{a=1}^{k} \sum_{i \in G_a} (w_i - \bar{w}_a)^2 
\end{equation}
By Lemma \ref{lm:tech3}, the first term of the right most hand side of (\ref{eq:condvar}) equals to
\[ 2\sum_{a=1}^{k} \left[ \sum_{b=1}^{k} n_{ba} \theta_b^2 - \frac{1}{n_a} \left( \sum_{b \in [k]} n_{ba} \theta_{b} \right)^2  \right] = \sum_{a=1}^{k} \sum_{u \neq v} \frac{n_{ua} n_{va}}{n_a} (\theta_u - \theta_v)^2 \]
The second term of of the right most hand side of (\ref{eq:condvar}) can be upper bounded by 
\[ 2\sum_{a=1}^{k} \sum_{i \in G_a} w_i^2 \le 2\sum_{i=1}^{n} w_i^2 \le 2.2 n \sigma_j^2  \]
on event $\mathcal{E}$, where the last inequality is due to Lemma \ref{lm:tech1}. Thus, we obtain
\[ c_j \le \sum_{u \neq v} \left( \sum_{a=1}^{k} \frac{n_{ua} n_{va}}{n_a} \right) (\theta_u - \theta_v)^2 + 2.2n \sigma_j^2 \]
on event $\mathcal{E}$ when $j \in S$. \\

Next, we lower bound the marginal variance $m_j$. Let $\bar{x} = \frac{1}{n} \sum_{i=1}^{n} x_i$, then we have 
\begin{equation*}
m_j = \sum_{i=1}^{n} (x_i - \bar{x})^2 = \sum_{i=1}^{n} (\theta_{z_i} - \bar{\theta} + w_i- \bar{w})^2 
\end{equation*}
Using the fact that $(u+v)^2 \ge \frac{1}{2} u^2 - 2v^2$, we obtain
\begin{equation*}
m_j \ge \frac{1}{2} \sum_{i=1}^{n} (\theta_{z_i}-\bar{\theta})^2 - 2 \sum_{i=1}^n (w_i - \bar{w})^2 \ge \frac{1}{2} \sum_{a=1}^{k} n_a^* (\theta_a-\bar{\theta})^2 - 2.2n \sigma_j^2.
\end{equation*}
Here the last inequality is due to $\sum_{i=1}^n (w_i - \bar{w})^2 \le \sum_{i=1}^n w_i^2 \le 1.1 n\sigma_j^2$ on event $\mathcal{E}$. Note that $\bar{\theta} = \frac{1}{n} \sum_{i=1}^{n} \theta_i = \frac{1}{n} \sum_{a=1}^{k} n_a^* \theta_a $. Lemma \ref{lm:tech3} implies
\[ n \sum_{a=1}^{k} n_a^* (\theta_a-\bar{\theta})^2 = \left( \sum_{a=1}^{k} n_a^* \right) \left( \sum_{a=1}^{k} n_a^* \theta_a^2 \right) - \left(\sum_{a=1}^{k} n_a^* \theta_a \right)^2 = \frac{1}{2} \sum_{u,v} n_u^* n_v^* (\theta_u - \theta_v)^2. \]
Since $B \le \frac{\alpha}{16}$, 
\[ \sum_{a=1}^{k} \frac{n_{ua} n_{va}}{n_a} = \sum_{a \neq v} \frac{n_{ua}}{n_a} n_{va} + \frac{n_{uv} n_{vv}}{n_v} \le B n_v^* +  B n_u^* \le \frac{1}{8n} n_u^* n_v^*  \]
Consequently,
\[ c_j \le \sum_{u \neq v} \left(\frac{n_{u}^* n_{v}^*}{8n} \right) (\theta_u - \theta_v)^2 + 2.2n \sigma_j^2 = \frac{1}{4} \sum_{a=1}^{k} n_a^* (\theta_a-\bar{\theta})^2 + 2.2n \sigma_j^2 \le 0.9 m_j, \]
provided $\textrm{SNR} \ge 21$. 
\end{proof}

\begin{proof}[Proof of Part (b)]
When $j \in S^c$, the conditional variance of variable $j$ can be simplified to
\[ c_j = \sum_{a=1}^{k} \sum_{i \in G_a} (w_i - \bar{w}_a)^2 = \sum_{a=1}^{k} \left( \sum_{i \in G_a} w_i^2 - n_a \bar{w}_a^2\right) = \sum_{i=1}^{n} w_i^2 - \sum_{a=1}^{k} \frac{1}{n_a} W_{G_a}^2 \]
Now we need an upper bound of $\sum_{a=1}^{k} \frac{1}{n_a} W_{G_a}^2$. The key difficulty is the possible dependence between the partition $G$ and $w_i$. When $\ell(G, T) \le \alpha/128$, we have the following lemma, whose proof is deferred to Section \ref{sec:tech}.
\begin{lemma} \label{lm:tech2}
There is a constant $c$ such that
\begin{equation} \label{eq:tech4}
\sum_{a=1}^{k} \frac{1}{n_a} W_{G_a}^2 \le 0.18 \sigma^2  n \;\; \text{ for all } G \text{ with } \ell(G, T) \le \alpha/128
\end{equation}
with probability greater than $1-\exp(-c \alpha n)$.
\end{lemma}

\noindent Then, Lemma $\ref{lm:tech2}$ and Lemma $\ref{lm:tech1}$ imply
$$\sum_{a=1}^{k} \frac{1}{n_a} W_{G_a}^2 \le 0.18 n \sigma^2 \le 0.2 \sum_{i=1}^{n} w_i^2$$ 
with probability greater than $1-\exp(-c\alpha n)$ for some constant $c$. Consequently, we have $c_j \le 0.8 \sum_{i=1}^{n} w_i^2$\\

From the proof of Theorem 3.2, when $j \in S^c$, we have
\begin{equation*}
m_j = \sum_{i=1}^{n} (w_i - \bar{w})^2 = \sum_{i=1}^{n} w_i^2 - n \bar{w}^2.
\end{equation*}
Since $ \sqrt{n} \bar{w_i}$ is a standard normal random variable, then $ \mathbb{P} \left\{ (\sqrt{n} \bar{w})^2 \ge 0.01 n \sigma^2 \right\} \le \exp(-c_1n) $
for some constant $c_1$. This, together with Lemma \ref{lm:tech1}, implies $$m_j \ge \sum_{i=1}^{n} w_i^2 - 0.01 n \sigma^2 \ge 0.98 \sum_{i=1}^{n} w_i^2.$$
The proof is complete.
\end{proof}

\section{Proof of Technical Lemmas}\label{sec:tech}

\begin{proof}[Proof of Lemma \ref{lm:sintheta}]
Lemma \ref{lm:sintheta} is essentially the Theorem 3 of \citep{cai2016rate}. Here we slightly modify their proof to obtain an in-probability upper bound. We first introduce some notations. For two $n \times k$ matrices $U$ and $\hat{U}$ with orthogonal columns, let $\sigma_1 \ge \sigma_2 \ge \cdots \ge \sigma_k \ge 0$ be the singular values of $U\hat{U}^{T}$. Then we define 
\[ \Theta(U,\hat{U}) = \textrm{diag}(\cos^{-1}(\sigma_1), \cdots, \cos^{-1}(\sigma_k)) \]
as the principal angles between $U$ and $\hat{U}$. And we use $\sin \Theta(U,\hat{U})$ to measure the distance between the column spaces of $U$ and $\hat{U}$. The $\sin \Theta$ distance has the following property. 
\begin{equation} \label{eq:sin1}   
\inf_{O \in \mathbb{O}_k} \| \hat{U} - U O \|_F \le \sqrt{2k} \| \sin \Theta(U, \hat{U}) \|.
\end{equation}
For any matrix $A \in \mathbb{R}^{n \times p}$, we denote $\mathbb{P}_A \in \mathbb{R}^{n \times n}$ as the projection matrix onto the column space of A.  Given the singular value decomposition of $A = UDV^{T}$ with $D$ non-singular, the projection matrix $\mathbb{P}_A$ equals to $U U^{T}$. To better present the results, we use a generic notation $C$ to denote absolute constants, whose value may vary from context to context. 

Now we are ready to prove the lemma. Without loss of generality, we assume $\sigma=1$. Otherwise, we can re-scale the signal and the noise matrix by $1/\sigma$. By Proposition 1 in \citep{cai2016rate}, we have
\[ \| \sin \Theta(\hat{U}, U) \|^2 \le \frac{\sigma_k^2(U^TY) \|U_\perp^{T} Y  \mathbb{P}_{U^TY}  \|^2 }{(\sigma_k^2(U^TY) - \sigma_{k+1}^2(Y))^2 }. \]
Following the proof of Theorem 3 in \citep{cai2016rate}, define the event $Q$ as
\begin{eqnarray*}  
Q &=& \left\{ \sigma_k^2(U^{T} Y) \ge \sigma_k^2(X) + p - \frac{1}{3} \sigma_k^2(X), \sigma_{k+1}^2(Y) \le p + \frac{1}{3} \sigma_k^2 (X),  \right. \\ 
&& \left. \|\mathbb{P}_{U^TY}\| \le \sqrt{\sigma_k^2(X)+p} \right\}. 
\end{eqnarray*}
Noting that $\| \sin \Theta(U, \hat{U}) \| \le 1$, the result is trivial when $\sigma_k^2(X) < C(\sqrt{np}+n)$ for some constant $C$.  Thus, it is sufficient to consider the case that $\sigma_k^2(X) \ge C(\sqrt{np}+n)$ for some large constant $C$, Lemma 4 in \citep{cai2016rate} gives us
\[ \mathbb{P} \left\{ Q^c \right\} \le C \exp \left( - \frac{C\sigma_k^4(X)}{\sigma_k^2(X) + p}  \right) \le \exp(- Cn). \]
On event $Q$, we have
\begin{equation} \label{eq:sin2} 
\| \sin \Theta(\hat{U}, U) \|^2 \le \frac{C(\sigma_k^2(X) + p)\|U_\perp^{T} Y  \mathbb{P}_{U^TY}  \|^2 }{\sigma_k^4(X)}. 
\end{equation}
Using Lemma 4 in \citep{cai2016rate} again, there exists a constant $C$ such that
\begin{equation} \label{eq:sin3} 
\mathbb{P} \left\{ \|U_\perp^{T} Y  \mathbb{P}_{U^TY}\| \ge C\sqrt{n} \right\} \le C \exp \left( - C n \right), 
\end{equation}
where we have used the fact that $p > n$. Combining the results of (\ref{eq:sin1}), (\ref{eq:sin2}) and (\ref{eq:sin3}), we obtain the desired result.
\end{proof}

\begin{proof}[Proof of Lemma \ref{lm:tech2}]
For a given $S \subseteq [n]$, $W_S = \sum_{i \in S} w_i$ is a Gaussian random variable with variance $\sigma^2 |S|$. Then $W_S^2$ is a sub-Exponential random variable satisfies 
$$\mathbb{E}\exp\left(\frac{\lambda W_S^2}{2\sigma^2 |S|}\right) \le \frac{1}{\sqrt{1-\lambda}}$$ for all $\lambda \in [0,1]$. Then by Chernoff bound, for a fixed partition $G$, we have
\[ \mathbb{P} \left\{ \sum_{a=1}^{k} \frac{1}{n_a} W_{G_a}^2 \ge t \right\} \le \mathbb{E} \exp \left( - \frac{\lambda}{2\sigma^2} + \sum_{a=1}^{k} \frac{\lambda W_{G_a}^2}{2\sigma^2 n_a} \right) \le \exp \left( - \frac{0.99 t}{2\sigma^2} + k \log 10 \right), \]
where we choose $\lambda = 0.99$ in the last inequality. By union bound,
\[ \mathbb{P} \left\{ \exists G \in \mathcal{G}, s.t. \sum_{a=1}^{k} \frac{1}{n_a} W_{G_a}^2 \ge t \right\} \le \exp \left( - \frac{0.99 t}{2\sigma^2} + k \log 10 + \log |\mathcal{G}|\right). \]
Now let us upper bound the cardinality of $\mathcal{G}$. First, there are 
\[ \prod_{a=1}^{k}{n_a^* \choose \gamma n_a^*} \le \prod_{a=1}^{k} \exp \left( \gamma n_a^* \log (e\gamma) \right) = \exp \left( \gamma n \log(e/\gamma) \right)  \]
possible choices of the elements that belongs to $\bigcup_{a=1}^{k} \left( T_a \cap G_a \right)$. For those at most $\gamma n$ elements that are not in $\bigcup_{a=1}^{k} \left( T_a \cap G_a \right)$, each of them have $k$ possible choices. Thus, the number of partitions $G$ is at most
\[ k^{\gamma n} \exp \left( \gamma n \log(e/\gamma) \right) = \exp(\gamma n \log(e k / \gamma)).  \]
Consequently, we obtain
\[ \mathbb{P} \left\{ \exists G \in \mathcal{G}, s.t. \sum_{a=1}^{k} \frac{1}{n_a} W_{G_a}^2 \ge 3\sigma^2 \gamma n \log (e k/ \gamma) \right\} \le \exp \left( - 0.4 \gamma n \log(ek/\gamma)\right). \]
Plug $\gamma = \frac{\alpha}{128}$ into above equality and note that $\gamma \log (ek/\gamma) \le \gamma\log(e/(\alpha \gamma)) \le 0.06$, the proof is complete.
\end{proof}

\end{document}